\documentclass[conference]{IEEEtran}
\usepackage{amsmath,epsfig,color,url}
\usepackage{amsfonts,bm}
\usepackage{graphicx}  
\usepackage[utf8]{inputenc}
\usepackage[T1]{fontenc}
\usepackage{float}
\usepackage[caption=false]{subfig}
\usepackage{wrapfig, epsfig}
\usepackage{psfrag}                                           
\usepackage{fancyhdr}
\usepackage{multirow}
\usepackage{mathrsfs}
\usepackage{amsmath,amsthm,amssymb}
\usepackage{float,balance}
\usepackage{tikz}
\usepackage{pgf}
\usetikzlibrary{dsp,chains}
\usetikzlibrary{arrows}
\newtheorem{theorem}{Theorem}

\newtheorem{lemma}{Lemma}

\DeclareMathOperator{\tr}{tr}
\newcommand\prob[1]{\mathop{\rm P}\nolimits\left[#1\right]}

\begin{document}
\IEEEoverridecommandlockouts
\title{Parallel BCC with One Common and Two Confidential Messages and Imperfect CSIT}

\author{ \IEEEauthorblockN{Ahmed Benfarah, Stefano Tomasin and Nicola Laurenti}
           \IEEEauthorblockA{Department of Information Engineering, University of Padova  \\
                             via Gradenigo 6/B, 35131 Padova, Italy. Email: firstname.lastname@dei.unipd.it }
                             \thanks{
This work was partly supported by the Italian Ministry of Education and
Research (MIUR) under project ESCAPADE
(Grant RBFR105NLC) in the ``FIRB-Futuro in Ricerca 2010'' funding program.
}
                             }
                             
\maketitle

\begin{abstract}

We consider a broadcast communication system over parallel sub-channels where the transmitter sends three messages: a common message to two users, and two confidential messages to each user which need to be kept secret from the other user. We assume partial channel state information at the transmitter (CSIT), stemming from noisy channel estimation. The first contribution of this paper is the characterization of the secrecy capacity region boundary as the solution of weighted sum-rate problems, with suitable weights. Partial CSIT is addressed by adding a margin to the estimated channel gains. The second paper contribution is the solution of this problem in an almost closed-form, where only two single real parameters must be optimized, e.g., through dichotomic searches. On the one hand, the considered problem generalizes existing literature where only two out of the three messages are transmitted. On the other hand, the solution finds also practical applications into the resource allocation of orthogonal frequency division multiplexing (OFDM) systems with both secrecy and fairness constraints. 
\end{abstract}


\section{Introduction}\label{sec:intro}

With the widespread adoption of wireless networks, security becomes an inherent issue of nowadays communications. In this context, \textit{physical layer security} arises as a promising tool to complement traditional cryptographic solutions. The basic concepts of this approach were founded by the pioneering work of Wyner~\cite{Wyner}. He introduced the \textit{wiretap} channel model in which the transmitter aims at sending reliably a confidential message to the legitimate receiver in presence of an  eavesdropper. The \textit{secrecy capacity} measures the maximum information rate at which the transmitter can reliably communicate a secret message to the receiver, while the eavesdropper left with no information on the message. Recently, the wiretap channel witnessed a renewed interest and many research works investigated the secrecy capacity of wireless fading~\cite{Gopala}, parallel~\cite{Laurenti,Renna12} and multiple-input multiple-output (MIMO) channels~\cite{Khistia,Renna14}. All of these works deal with the point-to-point wiretap channel model. There has been also an effort to generalize physical layer security to the multi-user context (see~\cite{Mukherjee} for a survey).

An important scenario of multi-user physical layer security is the \textit{broadcast channel with confidential messages} (BCC)~\cite{Csiszar}. In~\cite{Liang08}, the authors established the secrecy capacity region of parallel sub-channels where a source node has a common message for two receivers and a confidential message is intended only for one receiver. Extensive research work was made to characterize the secrecy capacity region of Gaussian MIMO BCC~\cite{Ekrem12}~\cite{Liu13}. In all these works, the communication scenario consists of a source node communicating with two receiving users maliciously eavesdropping on each other. Secure broadcasting to multiple receivers was analyzed in~\cite{Khisti08,Ekrem11} when the eavesdropper is external to the group of users. For an overview of the different considered BCC scenarios, the reader can see~\cite{Ekrem10}.

In this paper, we consider a parallel BCC with two receivers, where the transmitter aims at sending three independent messages with a total power constraint: one common message to both users and two confidential messages, one for each user. We further consider the case in which only partial channel state information at the transmitter (CSIT) is available before transmission, stemming from a noisy estimate of the channels. We first characterize the secrecy capacity region of the considered system where  partial CSIT is addressed by adding a margin to the estimated channel gains. Then, an almost closed-form solution to the weighted sum-rate maximization problem is derived, where two real variables must be optimized, e.g., through dichotomic search. Our contribution generalizes some related work which considered only two out of the three possible messages: in~\cite{Jorswieck,Ekrem09}, the authors derived the optimal power allocation in presence of two confidential messages without a common while in~\cite{Liang08}, the optimal power allocation for the case of one common message and one confidential message was established.


\section{System Model}\label{sec:model}

We consider\footnote{\textit{Notation:} Vectors and matrices are written in bold letters. $\log$ and $\ln$ denote the base-2 and natural-base logarithms, respectively. We indicate the positive part of a real quantity $x$ as $\left[x \right]^+ {=} \max \lbrace x ; 0 \rbrace $. $\mathbb{E}[X]$ denotes the expectation of the random variable $X$, $\mathbb{I}(X;Y)$ denotes the mutual information between variables $X$ and $Y$. $\tr(\bm{X})$ denotes the trace of a square matrix $\bm{X}$. For two positive semi-definite matrices $\bm{X}$ and $\bm{Y}$, we write $\bm{X} \preceq \bm{Y}$ whenever $\bm{Y}-\bm{X}$ is a positive semi-definite matrix.}  parallel BCC (e.g., OFDM) with $L$ sub-channels, one transmitter and two receiving users. Note that we consider real-valued signals. The transmitter sends the real-valued symbol $x_\ell$ on sub-channel $\ell$. The channel input is subject to the statistical total power constraint
\begin{equation}\label{eq:powerconstraint}
\sum_{\ell=1}^L \mathbb{E} \{x_\ell^2\} \leq P .
\end{equation}
We assume that the channel is quasi-static, i.e., it remains constant over the entire duration of a single packet.
At sub-channel $\ell$ of receiver $i =1,2$ we obtain
\begin{equation}\label{eq:model}
y_{i,\ell}= h_{i,\ell} x_{\ell} + n_{i,\ell}\,,
\end{equation}
where $n_{i, \ell}$ is the real-valued zero-mean unit variance additive white Gaussian noise (AWGN) term. Noise components for different sub-channels are independent. $h_{i,\ell}$ is the real-valued channel coefficient. Let $\alpha_{i,\ell} =  h_{i,\ell}^2$ be the channel power gain. We assume that the transmitter has some partial channel state information. It knows the channel statistical distribution and possesses estimates $\hat{h}_{i,\ell}$ of the channel coefficients, that are corrupted by noise
\begin{equation}\label{eq:estimate}
\hat{h}_{i,\ell} = h_{i,\ell} + \eta_{i,\ell} \, ,
\end{equation}
where $\eta_{i,\ell}$ are iid real-valued zero-mean Gaussian noise with variance $\sigma^2$. The conditional probability density function (pdf) of the channel gain $\alpha_{i,\ell}$ given the channel coefficient estimate $\hat{h}_{i,\ell}$ can be computed from the \textit{a priori} pdf of the channel coefficient $f_{h_{i,\ell}}$ and that of the estimation noise $f_\eta$ as
\begin{equation}\label{eq:bayes}
\begin{split}
f_{\alpha_{i,\ell} | \hat{h}_{i,\ell}}(a | b) = \\
 \frac{f_{\eta}(b-\sqrt a) f_{h_{i,\ell}}(\sqrt a) + f_{\eta}(b+\sqrt a) f_{h_{i,\ell}}(-\sqrt a)}{2\sqrt a \left [ f_{h_{i,\ell}} \otimes f_{\eta} \right](b) } \, .
\end{split}
\end{equation}
where $\otimes$ denotes the convolution operation. 

As illustrated in Fig.~\ref{fig:ofdmbroadcast}, we consider a BCC where the transmitter aims at reliably delivering a common message $M_0$ with information rate $R_0$ and two separate confidential messages $M_1$ and $M_2$ with information rates $R_1$ and $R_2$, respectively~\cite{Xu}. The common message $M_0$ is intended for both receivers, while  confidential message $M_i$ is intended for receiver $i$  and needs to be kept secret from the other receiver. The transmitter allocates power $p_{i,\ell}$ on sub-channel $\ell$ for the confidential message $M_i$, and power $p_{0, \ell}$ for the common message $M_0$. 

In order to have {\em reliable transmissions} to the intended receiver, vanishing error probabilities must be obtained as the codeword length $n$ grows to infinity. {\em Secrecy} is measured in terms of the information leakage rate to the non-intended receiver (aka \textit{weak} information theoretic secrecy)~\cite{Wyner,Csiszar}, i.e., defining $\bm{Y}_i^n = [\bm{y}_{i}(1), \ldots, \bm{y}_{i}(n)]$,  we require
\begin{equation}\label{eq:secrecycontraint}
\frac{1}{n} \mathbb{I}(M_1;\bm{Y}_2^n)\rightarrow 0 \quad ; \quad \frac{1}{n} \mathbb{I}(M_2;\bm{Y}_1^n)\rightarrow 0, 
\end{equation}
as $n \rightarrow \infty$.

\begin{figure}
\centering
\scalebox{0.8}{
\begin{tikzpicture}
	\node[circle] (Al) {};
	\node[dspfilter, right= of Al,xshift=4em] (tx) {Tx};
	\node[coordinate, right= of tx] (in) {};
	\node[coordinate, above= of in, yshift=-.2em] (in1) {};
	\node[coordinate, below= of in, yshift=.2em] (in2) {};
	\node[dspmixer, right= of in1] (ch1) {};
	\node[dspmixer, right= of in2] (ch2) {};
	\node[below= of ch1,yshift=1em,inner sep=0mm,outer sep=0mm] (h1) {$\{h_{1,\ell}\}$};
	\node[below= of ch2,yshift=1em,inner sep=0mm,outer sep=0mm] (h2) {$\{h_{2,\ell}\}$};
	\node[dspadder,right= of ch1] (a1) {};		
	\node[dspadder,right= of ch2] (a2) {};		
	\node[below= of a1,yshift=1em,inner sep=0mm,outer sep=0mm] (n1) {$\{n_{1,\ell}\}$};
	\node[below= of a2,yshift=1em,inner sep=0mm,outer sep=0mm] (n2) {$\{n_{2,\ell}\}$};		
	\node[dspfilter,right= of a1] (rx1) {\, Rx 1 };		
	\node[dspfilter,right= of a2] (rx2) {\,  Rx 2 };		
	\node[circle,right= of rx1] (Bo) {};
	\node[circle,right= of rx2] (Ev) {};
 	\draw[dspconn] (Al) -- node[midway,above] {$(M_0, M_1, M_2)$} (tx);
	\draw[dspline] (tx) -- node[midway,above] {$\{x_{\ell}\}$} (in); 
	\draw[dspline] (in) -- (in1); 	
	\draw[dspline] (in) -- (in2); 		
	\draw[dspconn] (in1) -- (ch1); 	
	\draw[dspconn] (in2) -- (ch2); 		
	\draw[dspconn] (ch1) -- (a1); 		
	\draw[dspconn] (ch2) -- (a2); 				
	\draw[dspconn] (a1) -- node[midway,above] {$\{y_{1,\ell}\}$} (rx1); 		
	\draw[dspconn] (a2) -- node[midway,above] {$\{y_{2,\ell}\}$} (rx2); 				
	\draw[dspconn] (rx1) -- node[midway,above] { $ \qquad (\hat{M}_0,\hat{M}_1)$} (Bo); 		
	\draw[dspconn] (rx2) -- node[midway,above] { $ \qquad (\tilde{M}_0,\hat{M}_2)$} (Ev); 				
 	\draw[dspconn] (n1) -- (a1);
 	\draw[dspconn] (n2) -- (a2);
	\draw[dspconn] (h1) -- (ch1); 	
	\draw[dspconn] (h2) -- (ch2); 	
	
\end{tikzpicture}}
\caption{Parallel BCC with common and two confidential messages.}\label{fig:ofdmbroadcast}
\end{figure}
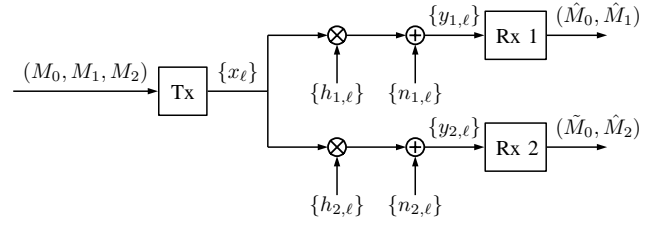

\subsection{Secrecy Capacity  Region}\label{subsec:secrecyregion}

Since the transmitter does not know the exact channel realization, secrecy outage may occur, i.e., the transmitted message is either not secret or not decoded by the receiver. However, computing the secrecy outage probability is an involved task, therefore we consider here a simpler approach where we add some margin to the channel estimates in order to keep the outage probability under control. 

In particular, the transmitter can compute upper and lower bounds on the channel gains $\alpha_{i,\ell}^+$ and $\alpha_{i,\ell}^-$ that provide the desired outage probability. We consider here a simpler approach where the same probability threshold $\varepsilon$ is used on each channel, i.e.,
\begin{equation}
\prob{\alpha_{i,\ell} > \alpha_{i,\ell}^+ \left| \hat{h}_{i,\ell}\right.} < \varepsilon \,,\quad\prob{\alpha_{i,\ell} < \alpha_{i,\ell}^- \left| \hat{h}_{i,\ell}\right.} < \varepsilon \, .  \\  \label{eq:bestworstgain}
\end{equation}
Then, $\alpha_{i,\ell}^-$ will be considered as the channel gain to the intended receiver, while $\alpha_{i,\ell}^+$ is the channel gain to the unintended receiver. The probabilities in (\ref{eq:bestworstgain}) can be computed using the pdf 
(\ref{eq:bayes}). Note also that when perfect CSIT is available $\alpha_{i,\ell}^- = \alpha_{i,\ell}^+$.


The secrecy capacity region $\mathcal{C}_s$ is defined as the closure of all rate triples $(R_0,R_1,R_2)$ that can be achieved by any coding scheme while maintaining both reliability and secrecy requirements. Let $\bm{K}_x$ be the covariance matrix of $\{x_\ell\}_{\ell=1,\ldots,L}$, the secrecy capacity region of the Gaussian MIMO BCC under a covariance constraint (i.e., $\bm{K}_x \preceq \bm{K}$) was characterized in~\cite{Ekrem12,Liu13}. The secrecy capacity region $\mathcal{C}_s$ under the  total power constraint (\ref{eq:powerconstraint}) is obtained by the union over all covariance matrices which satisfy $\tr(\bm{K}) \leq P$. Moreover, the parallel channels can be seen as a special case of MIMO channels. Having independent inputs for each sub-channel is optimal for the parallel BCC~\cite{Liang08}. Consequently, it is sufficient to have diagonal input covariance matrices~\cite{Liu13,Ekrem12} for the parallel BCC. We define the power allocation vector $\bm{p} =[p_{0, 1}, \ldots, p_{0,L}, p_{1, 1}, \ldots, p_{1, L}, p_{2,1}, \ldots, p_{2,L}]$. The set $\mathcal{P}$ includes all power allocation vectors $\bm{p}$ that satisfy the total power constraint (\ref{eq:powerconstraint}), i.e.,
\begin{equation}\label{eq:powerconstraintversiontwo}
\mathcal{P}= \left\{ \bm{p}: \sum_{\ell=1}^L (p_{0,\ell}+p_{1,\ell}+p_{2,\ell}) \leq P \right\} .
\end{equation}
Let us partition the channel index set $\{1,\ldots,L\}$ into 
\begin{equation}\label{eq:setsdefinition}
\begin{split}
 & S_1= \{\ell: \alpha_{1,\ell}^- > \alpha_{2,\ell}^+ \} \ , \quad S_2= \{\ell: \alpha_{2,\ell}^- >  \alpha_{1,\ell}^+ \} \ ,\\
 & S_3 = \{1,\ldots,L\} \setminus \left(S_1\cup S_2\right). 
\end{split}
\end{equation}  
By combining the results of~\cite{Ekrem12} and~\cite{Liang08}, the secrecy capacity region of the parallel BCC with common and two confidential messages can be written as

\begin{equation}\label{eq:secrecycapacityofdm}
\mathcal{C}_s  =   \bigcup_{\bm{p} \in \mathcal{P}}  \left\{ \begin{array}{l}  (R_0,R_1,R_2): \\
                                     0 \leq R_0 \leq R_{0}^{\rm max}(\bm{p}) \\
                                                            \\
                                    0 \leq  R_i \leq R_{i}^{\rm max}(\bm{p}) \end{array} \right.   
\end{equation}
where 
\begin{equation}\label{eq:commonrate}
R_{0}^{\rm max}(\bm{p}) = \min \lbrace R_{01}^{\rm max}(\bm{p}),R_{02}^{\rm max}(\bm{p}) \rbrace
\end{equation}
\begin{equation}\label{eq:commonrateone}
\begin{split}
R_{01}^{\rm max}(\bm{p}) =  \frac{1}{2} \sum_{\ell=1}^L & \log \bigl(1+ \alpha_{1,\ell}^-[p_{0,\ell}+p_{1,\ell}+p_{2,\ell}]\bigr)\\
                                      & - \log \bigl(1+\alpha_{1,\ell}^-[p_{1,\ell}+p_{2,\ell}]\bigr)
\end{split}
\end{equation}  
\begin{equation}\label{eq:commonratetwo}
\begin{split}
R_{02}^{\rm max}(\bm{p}) =  \frac{1}{2} \sum_{\ell=1}^L & \log \bigl(1+ \alpha_{2,\ell}^-[p_{0,\ell}+p_{1,\ell}+p_{2,\ell}]\bigr)\\
                                      & - \log \bigl(1+\alpha_{2,\ell}^-[p_{1,\ell}+p_{2,\ell}]\bigr)
\end{split}
\end{equation} 
\begin{equation}\label{eq:confidentialone}
R_1^{\rm max}(\bm{p}) = \frac{1}{2}  \sum_{\ell \in S_1} \log(1+\alpha_{1,\ell}^- p_{1,\ell}) -\log(1+  \alpha_{2,\ell}^+ p_{1,\ell})
\end{equation}
\begin{equation}\label{eq:confidentialtwo}
R_2^{\rm max} (\bm{p}) = \frac{1}{2}  \sum_{\ell \in S_2} \log(1+\alpha_{2,\ell}^- p_{2,\ell}) -\log(1+  \alpha_{1,\ell}^+ p_{2,\ell}) .
\end{equation}
The expression of the secrecy capacity region states that the receivers decode the common message first, by treating the confidential messages as noise. Then, each receiver decodes its own confidential message.
 
Note that the expressions of the secrecy capacity region hold also for perfect CSIT by letting $\alpha_{i,\ell}^- = \alpha_{i,\ell}^+$.

 \section{ Power Allocation Algorithm}\label{sec:algorithm}

First note that the secrecy capacity region  (\ref{eq:secrecycapacityofdm}) is convex. Therefore, for each triplet $(R_0^{\rm max}(\bm{p}^*),R_1^{\rm max}(\bm{p}^*),R_2^{\rm max}(\bm{p}^*))$ on the region boundary, there exists a weight triplet $w_0, w_1, w_2 >0$ satisfying
\begin{equation}\label{eq:optimizationproblem}
\begin{split}
\bm{p}^*= \arg \max_{\bm{p} \in \mathcal{P}  } & \left[  w_0 R_{0}^{\rm max}(\bm{p})   + w_1 R_{1}^{\rm max}(\bm{p}) + w_2 R_{2}^{\rm max} (\bm{p}) \right] .
\end{split}
\end{equation}
By solving (\ref{eq:optimizationproblem}) all points of the secrecy capacity region boundary are obtained. Note also that the weighted sum rate problem is also of interest for resource allocation in OFDM systems with a fairness constraint, where the weights are selected in order to enforce the desired fairness.   

Now, the optimization problem (\ref{eq:optimizationproblem}) together with (\ref{eq:commonrate}) is a max-min optimization, and can be solved by using an approach similar to that of~\cite{Liang07}. The particular result is provided in the following lemma, whose proof is not reported as it follows the same steps of~\cite{Liang07}.

\begin{lemma}\label{lem:maxmin}
The solution of (\ref{eq:optimizationproblem}) also solves one of the following three problems:
%
\begin{IEEEeqnarray*}{rCl}
\mbox{\bf(P1)} \ \bm{p}^{(1)} = & \arg \max_{\bm{p} \in \mathcal{P}} & \left[  w_0 R_{01}^{\rm max} (\bm{p}) + w_1 R_{1}^{\rm max}  (\bm{p}) 
{} + w_2 R_{2}^{\rm max} (\bm{p}) \right] \\
\mbox{\bf(P2)} \ \bm{p}^{(2)} = & \arg \max_{\bm{p} \in \mathcal{P}} & \left[  w_0 R_{02}^{\rm max} (\bm{p}) + w_1 R_{1}^{\rm max}  (\bm{p}) 
{} + w_2 R_{2}^{\rm max} (\bm{p}) \right] \\
\mbox{\bf(P3)} \ \bm{p}^{(3)} = & \arg \max_{\bm{p} \in \mathcal{P}} & \left[  w_0 (\mu R_{01}^{\rm max} (\bm{p}) + (1-\mu)R_{02}^{\rm max} (\bm{p}))  \right. \\ & &\left. {} + w_1 R_{1}^{\rm max}  (\bm{p}) + w_2 R_{2}^{\rm max} (\bm{p}) \right] 
\end{IEEEeqnarray*} 
for some $\mu \in (0,1)$ in (P3). In particular, 
\begin{equation}\label{eq:caseslemma}
\bm{p}^*=
\begin{cases}
 \bm{p}^{(1)} & \rm{if} \; R_{01}^{\rm max}(\bm{p}^{(1)}) < R_{02}^{\rm max}(\bm{p}^{(1)})  \\ 
    \bm{p}^{(2)} &  \rm{if} \;  R_{01}^{\rm max}(\bm{p}^{(2)}) > R_{02}^{\rm max} (\bm{p}^{(2)}) \\
   \bm{p}^{(3)}  &  \rm{if}  \; R_{01}^{\rm max}(\bm{p}^{(3)}) = R_{02}^{\rm max} (\bm{p}^{(3)}) 
\end{cases}
.
\end{equation}

\end{lemma}

We now focus on the solution of problems (P1)-(P3). Before introducing the result, we define the following terms, with $\ell =1,\ldots,L$ and $i = 1,2$. Let $ \bar{\imath}  = 2$ if $i=1$ and $\bar{\imath} = 1$ if $i =2$, let $ \mu_i = \mu $ if $i =1$ and $\mu_i = 1- \mu$ if $i =2$. Denote by $\delta_{i,\ell} =1/\alpha_{\bar{\imath},\ell}^+{-}1/\alpha_{i,\ell}^-$, and let $\lambda \geq 0 $ be a real valued parameter. Then, denote
\begin{subequations}
\begin{equation}
\beta_{i,\ell} =\frac{1}{2} \left[\delta_{i,\ell} \left(\delta_{i,\ell}+\frac{2 w_i}{\lambda \ln 2}\right)\right]^{1/2} -\frac{1}{2}\left(\frac{1}{\alpha_{\bar{\imath},\ell}^+}+\frac{1}{\alpha_{i,\ell}^-}\right)  \\ \label{eq:squareroot}
\end{equation}
\begin{equation}
\gamma_{i,\ell}= \frac{w_0}{2 \lambda \ln 2} - \frac{1}{\alpha_{i,\ell}^-} \,, \quad 
\zeta_{i,\ell} = \frac{w_i}{w_0} \delta_{i,\ell} - \frac{1}{\alpha_{\bar{\imath},\ell}^+} \\ \label{eq:intersectioncaseone}
\end{equation}
\begin{align}\label{eq:rootcasethree}
\nu_{i,\ell} & =  \frac{1}{2} \left[  \left( \frac{1}{\alpha_{\bar{\imath},\ell}^{-}} -\frac{1} {\alpha_{i,\ell}^{-}} -\frac{w_0}{2\lambda \ln 2} \right)^2 \right. \notag  \\
            & \left. \qquad \qquad + \frac{2 w_0 \mu_i}{ \lambda \ln 2} ( \frac{1}{\alpha_{\bar{\imath},\ell}^{-}} - \frac{1} {\alpha_{i,\ell}^{-} } )  \right]^{1/2} \notag \\
             & \quad-  \frac{1}{2} \left(\frac{1}{\alpha_{\bar{\imath},\ell}^{-}}+\frac{1}{\alpha_{i,\ell}^{-}}-\frac{w_0}{2 \lambda \ln 2} \right) \notag \\
\end{align}
\begin{align*}
\Delta_{i,\ell} = \left[ \left(\frac{w_i}{w_0}\right)^2 + \frac{w_i}{w_0} \left( \frac{2 \cdot (\frac{2}{\alpha_{\bar{\imath},\ell}^{-}}-\frac{1}{\alpha_{i,\ell}^{-}}- \frac{1} {\alpha_{\bar{\imath},\ell}^{+}})}{\frac{1}{\alpha_{\bar{\imath},\ell}^{+}}-\frac{1}{\alpha_{i,\ell}^{-}} } \right) + 1 \right] \\
\cdot \left[ \frac{1}{\alpha_{\bar{\imath},\ell}^{+}} - \frac{1}{\alpha_{i,\ell}^{-}} \right]^2
\end{align*}
\begin{equation}\label{eq:intersectionnew}
\theta_{i,\ell} = \frac{ \frac{w_i}{w_0} \bigl( \frac{1} {\alpha_{\bar{\imath},\ell}^{+}}- \frac{1} {\alpha_{i,\ell}^{-}} \bigr)- \bigl( \frac{1} {\alpha_{\bar{\imath},\ell}^{+}} + \frac{1} {\alpha_{i,\ell}^{-}} \bigr) + \sqrt{\Delta_{i,\ell}} } {2}  \\
\end{equation}
\begin{align*}\label{eq:deltapthree}
\Lambda_{i,\ell}  =  \bigl(\delta_{i,\ell}\bigr)^2 \left( \frac{w_i}{w_0} \right)^2 + 2 \frac{w_i}{w_0} \left[ \delta_{i,\ell} \bigl( \frac{2-\mu_i}{\alpha_{\bar{\imath},\ell}^{-}} - \frac{\mu_{\bar{\imath}}}{\alpha_{i,\ell}^{-}} - \frac{1}{\alpha_{\bar{\imath},\ell}^{+}} \bigr) \right] \\
                    \:  +  \bigl(\delta_{i,\ell}\bigr)^2 + \mu_i \bigl( \frac{1}{\alpha_{\bar{\imath},\ell}^{-}} - \frac{1}{\alpha_{i,\ell}^{-}} \bigr) \left[ \mu_i \bigl( \frac{1}{\alpha_{\bar{\imath},\ell}^{-}} -  \frac{1}{\alpha_{i,\ell}^{-}} \bigr) - 2 \bigl (  \frac{1}{\alpha_{\bar{\imath},\ell}^{+}} -  \frac{1}{\alpha_{\bar{\imath},\ell}^{-}} \bigr ) \right ] \\
\end{align*}
\begin{equation}
\xi_{i,\ell} = \frac{ \frac{w_i}{w_0} \delta_{i,\ell} - \bigl( \frac{1}{\alpha_{\bar{\imath},\ell}^{+}}+ \frac{1}{\alpha_{i,\ell}^{-}} \bigr) - \mu_i \bigl( \frac{1}{\alpha_{\bar{\imath},\ell}^{-}} - \frac{1}{\alpha_{i,\ell}^{-}} \bigr) + \sqrt{\Lambda_{i,\ell}} }{2}\ . \\ \label{eq:intersectioncasethree}
\end{equation}
\end{subequations}

The main result for the solution of the optimization problem (\ref{eq:optimizationproblem}) is provided by the following theorem.
\begin{theorem}\label{theo:powerallocation}
The solutions of problems (P1)-(P3) are:

\textbf{(P1)} For $\ell \in S_1$, if $\frac{w_1}{w_0} > \frac{\alpha_{1,\ell}^-}{\alpha_{1,\ell}^{-}-\alpha_{2,\ell}^+}$, then
\begin{subequations}\label{eq:p1}
\begin{equation}
p_{0,\ell}^{(1)}=\left[\gamma_{1,\ell}-\zeta_{1,\ell} \right]^+\,, \;   p_{1,\ell}^{(1)}= \left[  \min \left\lbrace \beta_{1,\ell};\zeta_{1,\ell} \right\rbrace \right]^+ . 
\label{p11} 
\end{equation}
Otherwise, if $\frac{w_1}{w_0} \leq \frac{\alpha_{1,\ell}^-}{\alpha_{1,\ell}^{-}-\alpha_{2,\ell}^+}$, then
\begin{equation}
p_{0,\ell}^{(1)} = \left[\gamma_{1,\ell} \right]^+ \,, \quad p_{1,\ell}^{(1)}=0 \:. 
\label{p12}
\end{equation}
For $\ell \in S_2$, if $\frac{w_2}{w_0} > \frac{\alpha_{1,\ell}^-}{\alpha_{2,\ell}^{-}- \alpha_{1,\ell}^{+} }$, then
\begin{equation}
p_{0,\ell}^{(1)}=\left[\gamma_{1,\ell}- \theta_{2,\ell} \right]^+\,,  \;  p_{2,\ell}^{(1)}= \left[  \min \left\lbrace \beta_{2,\ell}; \theta_{2,\ell}  \right\rbrace \right]^+ . 
\label{p13}
\end{equation}
Otherwise, if $\frac{w_2}{w_0} \leq \frac{\alpha_{1,\ell}^-}{\alpha_{2,\ell}^{-}- \alpha_{1,\ell}^{+} }$, then
\begin{equation}\label{eq:powerallocationcaseone}
p_{0,\ell}^{(1)} = \left[\gamma_{1,\ell} \right]^+\,,   \quad p_{2,\ell}^{(1)}=0  \, .
\end{equation}
For $\ell \in S_3$,
\begin{equation}
p_{0,\ell}^{(1)} = \left[ \gamma_{1,\ell} \right]^+ 
\label{p18}
\end{equation}
\end{subequations}
where $\lambda$ is chosen to satisfy the total power constraint (\ref{eq:powerconstraintversiontwo}).

\textbf{(P2)} Due to the symmetry (with respect to the user index) of problems (P1) and (P2), solution (P2) is the same as that of (P1) where user indices 1 and 2 are swapped.
%

\textbf{(P3)} For $\ell \in S_i$, if $\Lambda_{i,\ell} > 0$, then 

if $\frac{w_i}{w_0} > \frac{\mu_i \alpha_{i,\ell}^{-} + \mu_{\bar{\imath}} \alpha_{\bar{\imath},\ell}^{-}}{\alpha_{i,\ell}^{-} -\alpha_{\bar{\imath},\ell}^{+}}$, then
\begin{subequations}\label{eq:p3}
\begin{equation}
p_{0,\ell}^{(3)}= \left[ \nu_{i,\ell} - \xi_{i,\ell} \right]^+\,, \quad p_{i,\ell}^{(3)}= \left[  \min \left\lbrace \beta_{i,\ell};\xi_{i,\ell} \right\rbrace  \right]^+ .
\end{equation}
Otherwise, if $\frac{w_i}{w_0} \leq \frac{\mu_i \alpha_{i,\ell}^{-} + \mu_{\bar{\imath}} \alpha_{\bar{\imath},\ell}^{-}}{\alpha_{i,\ell}^{-} -\alpha_{\bar{\imath},\ell}^{+}}$, then
\begin{equation}
p_{0,\ell}^{(3)}= \left[ \nu_{i,\ell} \right]^+\,, \quad p_{i,\ell}^{(3)}=0 \, .
\end{equation}
If $\Lambda_{i,\ell} = 0 $, then 

if $\frac{w_i}{w_0} > \frac{ \alpha_{i,\ell}^{-} + \mu_{\bar{\imath}} \alpha_{\bar{\imath},\ell}^{+} + \mu_i \frac{\alpha_{i,\ell}^{-} \alpha_{\bar{\imath},\ell}^{+}}{\alpha_{\bar{\imath},\ell}^{-}}}{\alpha_{i,\ell}^{-} - \alpha_{\bar{\imath},\ell}^{+} }$, then
\begin{equation}
p_{0,\ell}^{(3)}=\left[\nu_{i,\ell} - \xi_{i,\ell} \right]^+\,, \quad p_{i,\ell}^{(3)}= \left[  \min \left\lbrace \beta_{i,\ell};\xi_{i,\ell} \right\rbrace \right] ^+ . 
\end{equation}
Otherwise, if $\frac{w_i}{w_0} \leq \frac{ \alpha_{i,\ell}^{-} + \mu_{\bar{\imath}} \alpha_{\bar{\imath},\ell}^{+} + \mu_i \frac{\alpha_{i,\ell}^{-} \alpha_{\bar{\imath},\ell}^{+}}{\alpha_{\bar{\imath},\ell}^{-}}}{\alpha_{i,\ell}^{-} - \alpha_{\bar{\imath},\ell}^{+} } $, then
\begin{equation}
p_{0,\ell}^{(3)} = \left[ \nu_{i,\ell} \right]^+ \, , \quad p_{i,\ell}^{(3)}=0  \, .
\end{equation}
If $ \Lambda_{i,\ell} < 0$, then 

if $ \frac{w_i}{w_0} > \frac{\mu_i \alpha_{i,\ell}^{-} + \mu_{\bar{\imath}} \alpha_{\bar{\imath},\ell}^{-}}{\alpha_{i,\ell}^{-} -\alpha_{\bar{\imath},\ell}^{+}} $, then
\begin{equation}
p_{0,\ell}^{(3)} = 0 \,, \quad p_{i,\ell}^{(3)}= \left[ \beta_{i,\ell} \right]^+  \, .
\end{equation}
Otherwise, if $ \frac{w_i}{w_0} \leq \frac{\mu_i \alpha_{i,\ell}^{-} + \mu_{\bar{\imath}} \alpha_{\bar{\imath},\ell}^{-}}{\alpha_{i,\ell}^{-} -\alpha_{\bar{\imath},\ell}^{+}}$, then 
\begin{equation}
p_{0,\ell}^{(3)} = \left[\nu_{i,\ell} \right]^+ \,, \quad p_{i,\ell}^{(3)}=0 \, .
\end{equation}
For $\ell \in S_3$, 
\begin{equation}
p_{0,\ell}^{(3)} = \left[\nu_{1,\ell} \right]^+
\end{equation}
\end{subequations}
where $\lambda$ is chosen to satisfy the total power constraint (\ref{eq:powerconstraintversiontwo}), and $\mu$ is chosen to satisfy $R_{01}^{\rm max}\bigl(\bm{p}^{(3)}\bigr)  = R_{02}^{\rm max}\bigl(\bm{p}^{(3)}\bigr)$.

\end{theorem}

\begin{proof}\label{pro:demonstrationalgorithm}
See the Appendix.
\end{proof}

Table~\ref{tab:powerallocation} summarizes the power allocation algorithm solving (\ref{eq:optimizationproblem}). 
The algorithm includes three steps consisting of simple closed-form solutions of problems (P1)-(P3). Steps 1 and 2 require the optimization of $\lambda$, while Step 3 requires the optimization of both $\lambda$ and $\mu \in (0,1)$.  As  $R_{01}^{\rm max}(\bm{p}^{(3)})$ and $R_{02}^{\rm max}(\bm{p}^{(3)})$ are monotonous functions versus both $\lambda$ and $\mu$, these optimizations can be performed efficiently for instance by a dichotomic search. Moreover, the two searches can be performed in cascade. 

\bigskip
\begin{table}
\centering
\caption{Power allocation algorithm.}\label{tab:powerallocation}
\scalebox{1.1}{
\begin{tabular}[ht]{rl}
\hline
\textit{Step 1} & Compute $\bm{p}^{(1)}$ by (\ref{eq:p1}). \\
       & \: If $R_{01}^{\rm max}(\bm{p}^{(1)}) {<} R_{02}^{\rm max}(\bm{p}^{(1)})$, then $\bm{p}^* = \bm{p}^{(1)}$ . \\
       & \: Otherwise, go to \textit{Step 2}. \\
       & \\
\textit{Step 2} & Compute $\bm{p}^{(2)}$ by (\ref{eq:p1}) with user indices exchanged. \\
                & \: If $R_{01}^{\rm max}(\bm{p}^{(2)}) {>} R_{02}^{\rm max}(\bm{p}^{(2)})$ , then $\bm{p}^* = \bm{p}^{(2)}$ . \\
                & \: Otherwise, go to \textit{Step 3}. \\
                &                                      \\
\textit{Step 3} & Compute $\bm{p}^{(3)}$ by (\ref{eq:p3}). \\
                & \: Then $\bm{p}^* =\bm{p}^{(3)}$ . \\
\hline
\end{tabular}} 
\end{table}
                 
Note that the power allocation algorithm can be used also for perfect CSIT by letting $\alpha_{i,\ell}^- = \alpha_{i,\ell}^+$.

\section{Numerical Results}\label{sec:simulations}

In this section, we validate the analytical results by considering a system where the number of sub-channels $L$ and the total power $P$ are both fixed to $64$. Each sub-channel is Rayleigh fading, thus, the powers of the sub-channel gains $h_{1,\ell}^2$ and $h_{2,\ell}^2$ are exponentially distributed with means $\rm {SNR}_1 = \mathbb{E}\{h_{1,\ell}^2\}$ and $\rm{SNR}_2 = \mathbb{E}\{h_{2,\ell}^2\}$.

Fig.~\ref{fig:secrecy_capacity} shows a contour plot of the boundary surface for the three dimensional secrecy capacity averaged over the channel realizations with $\rm{SNR}_1 = \rm{SNR}_2 = 10$ dB and perfect CSIT. We  remark that the surface of the secrecy capacity  gets smaller as $\mathbb{E}[R_0]$ increases. Moreover, the secrecy capacity region is symmetric for the same average $\rm{SNR}$ values of both users.

\begin{figure}
\centering
\psfrag{R1 [bits/s/Hz]}[t][]{\scriptsize $\mathbb{E} [R_1]$ [bits/s/Hz]}
\psfrag{R2 [bits/s/Hz]}[b][]{\scriptsize $\mathbb{E} [R_2]$ [bits/s/Hz]}
\psfrag{R0=1}[r][]{\scriptsize $\mathbb{E}[R_0] = 1$}
\psfrag{R0=0.8}[r][]{\scriptsize $\mathbb{E}[R_0] = 0.8$}
\psfrag{R0=0.6}[r][]{\scriptsize $ \mathbb{E} [R_0] = 0.6$}
\psfrag{R0=0.4}[][]{\scriptsize $\mathbb{E} [R_0] = 0.4$}
\includegraphics[height=0.9\hsize]{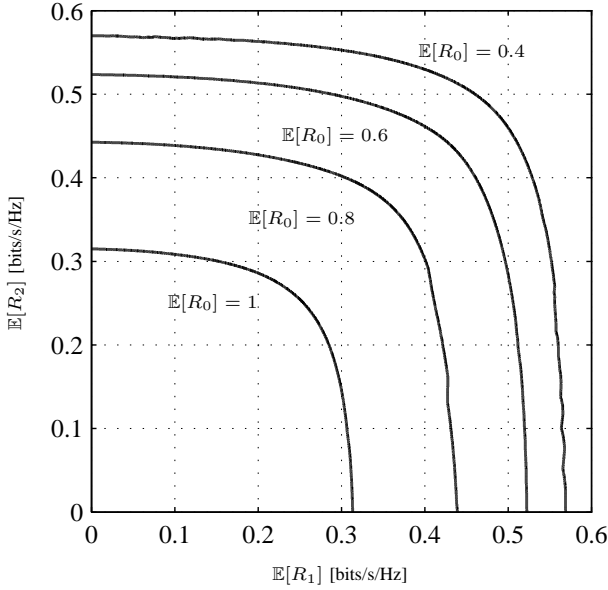}
\caption{Contour plot of the boundary surface of the secrecy capacity region.}
\label{fig:secrecy_capacity}
\end{figure}  
%
%

We then compare the algorithm with two other schemes. The first one is uniform power allocation over the sub-channels and over the three messages. The second scheme is the power allocation algorithm proposed in~\cite{Jorswieck} which maximizes the sum secrecy-rate in the presence of two confidential messages but without a common message. In order to compare~\cite{Jorswieck} with our approach, we first assign power $P/3$ to transmit the common message and then we split the remaining power $2\, P/3$ between the two confidential messages according to the algorithm of~\cite{Jorswieck}. Fig.~\ref{fig:comparisonalgorithms}.a compares the average weighted sum-rate (with $w_0=w_1=w_2$) of our algorithm with the average sum-rate of the two schemes versus $\rm{SNR} = \rm{SNR}_1 = \rm{SNR}_2$. The optimal algorithm provides a significant advantage mainly at high $\rm{SNR}$ range.

\begin{figure}
\centering
\psfrag{SNR [dB]}[t][t]{\scriptsize $\rm {SNR}$ [dB]}
\psfrag{Sum rate [bits/s/Hz]}[][]{\scriptsize Average weighted sum-rate [bits/s/Hz]}
\psfrag{Optimal}[][]{\scriptsize Optimal}
\psfrag{Uniform}[][]{\scriptsize Uniform}
\psfrag{Jors}[l][]{\scriptsize ~\cite{Jorswieck}}
\includegraphics[height=6.2cm,width=0.49\hsize]{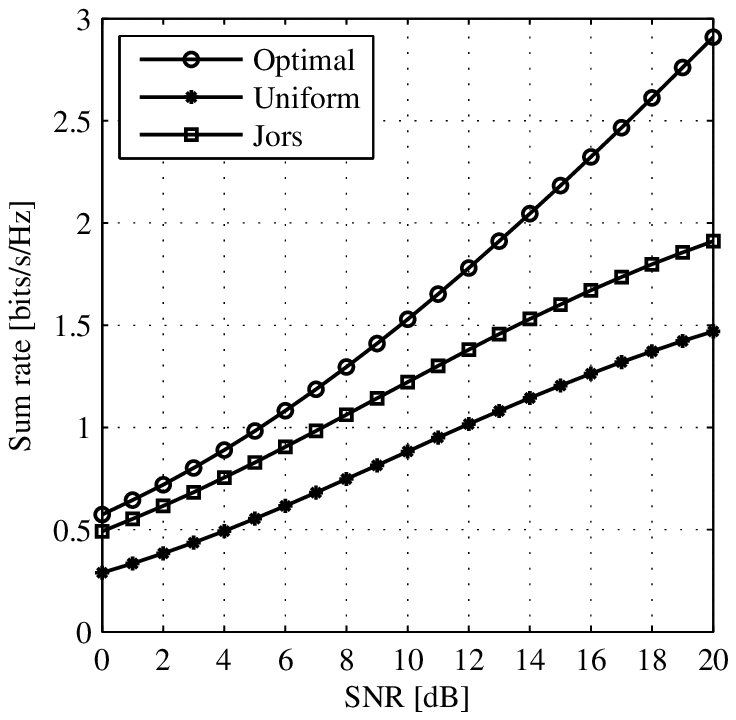}
\includegraphics[height=6.3cm,width=0.49\hsize]{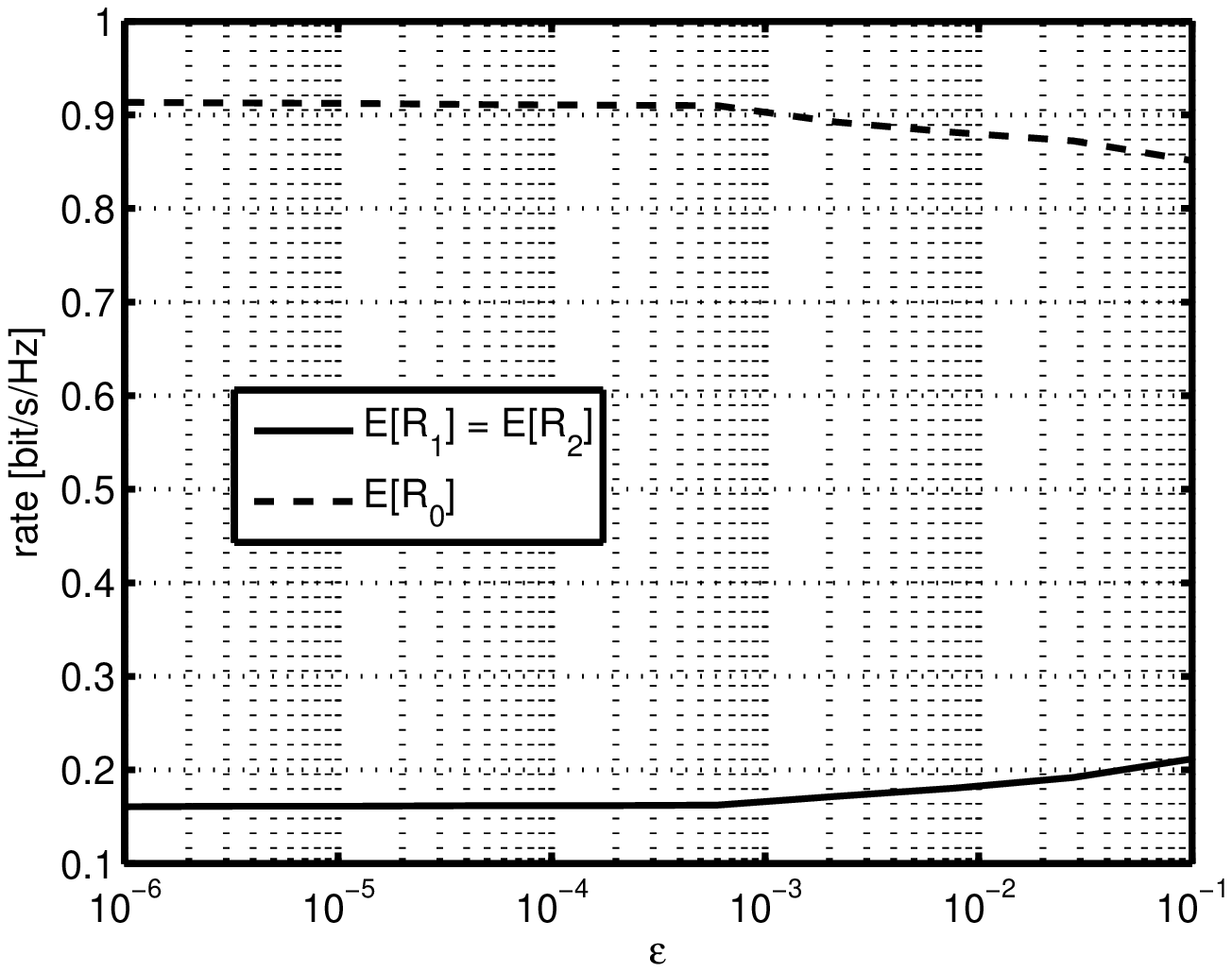}

a) \hspace{4cm} b)
\caption{a) Average weighted sum-rate ($w_0=w_1=w_2=1$) versus $\rm{SNR}=\rm{SNR}_1=\rm{SNR}_2$ of some power allocation algorithms.}\label{fig:comparisonalgorithms}
\end{figure}

Lastly, in Fig.~\ref{fig:comparisonalgorithms}.b  we show the rates when imperfect CSIT is available with $\sigma = 0.01$, as function of $\epsilon$. We note that as $\epsilon$ increases the secret rates increases, as we are less restrictive on the illegitimate channel, while on the other hand the rate of the common message decreases.

\section{Conclusions}\label{section:conclusion}

We have provided a characterization of the boundary for the secrecy capacity region of a parallel BCC with two confidential and one common messages by converting it into a power allocation problem. An almost closed-form solution to the problem has been provided, which can also be exploited in practical scenarios for  resource allocation in OFDM systems with secrecy and fairness constraints. Comparison with existing power allocation schemes highlights the significant advantage of the optimal algorithm.

\appendix
 We solve problems (P1)-(P3) by a technique based on deriving an upper bound on the Lagrangian operator and establishing power allocations that achieve the upper bound.

\textbf{(P1)}
The Lagrangian $\mathcal{L}$ of (P1) is given by
\begin{align}\label{eq:lagrangiancaseone}
\mathcal{L} & = \sum_{\ell=1}^L \frac{w_0}{2} \log \bigl(1+ \frac{ \alpha_{1,\ell}^{-} p_{0,\ell}}{1+ \alpha_{1,\ell}^{-} [p_{1,\ell}+p_{2,\ell}]} \bigr)  \nonumber \\
             & + \sum_{\ell \in S_1} \frac{ w_1}{2} \log \bigl( 1+ \alpha_{1,\ell}^{-} p_{1,\ell} \bigr) - \frac{w_1}{2} \log \bigl(1+ \alpha_{2,\ell}^{+} p_{1,\ell} \bigr) \nonumber \\
             & + \sum_{\ell \in S_2} \frac{w_2}{2} \log \bigl(1+ \alpha_{2,\ell}^- p_{2,\ell} \bigr) - \frac{w_2}{2} \log \bigl(1+ \alpha_{1,\ell}^{+} p_{2,\ell} \bigr) \nonumber \\
             & - \lambda \sum_{\ell=1}^L \bigl[p_{0,\ell} + p_{1,\ell} + p_{2,\ell}\bigr]
\end{align}
where $\lambda \geq 0$ is the Lagrange multiplier. For $\ell \in S_1$, the transmitter merely sends the common and the confidential messages $M_1$ (i.e., $p_{2,\ell}{=}0$). For $\ell \in S_2$, the transmitter sends the common and the confidential message $M_2$, (i.e. $p_{1,\ell}{=}0$). While, for $\ell \in S_3$, the transmitter sends only the common message. For $\ell \in S_1$, $p_{0,\ell}^{(1)}$ and $p_{1,\ell}^{(1)}$ need to maximize :
\begin{align}\label{eq:lagrangianonesone}
\mathcal{L}_1 & = \frac{w_0}{2} \log \bigl(1+\frac{ \alpha_{1,\ell}^- p_{0,\ell}} {1+ \alpha_{1,\ell}^- p_{1,\ell}} \bigr) + \frac{w_1}{2}  \log \bigl( 1+ \alpha_{1,\ell}^- p_{1,\ell} \bigr)  \nonumber \\
               & - \frac{w_1}{2}  \log \bigl(1+ \alpha_{2,\ell}^+ p_{1,\ell} \bigr)  - \lambda (p_{0,\ell} + p_{1,\ell}) .
\end{align}
We denote by $u_{0,\ell}(\cdot)$ and $u_{1,\ell}(\cdot)$ the partial derivative of $\mathcal{L}_1$ with respect to $p_{0,\ell}$ and $p_{1,\ell}$, respectively:
\begin{equation}\label{eq:fractioncommoncaseone}
u_{0,\ell}(x)= \frac{w_0}{2 \ln 2} \frac{\alpha_{1,\ell}^-}{1+\alpha_{1,\ell}^- x} - \lambda 
\end{equation}
\begin{equation}\label{eq:fractionconfidentialone}
u_{i,\ell}(x)= \frac{w_i}{2 \ln 2} \left( \frac{\alpha_{i,\ell}^-}{1+\alpha_{i,\ell}^- x}-\frac{\alpha_{\bar{\imath},\ell}^+}{1+\alpha_{\bar{\imath},\ell}^+ x} \right) - \lambda .
\end{equation}
Then, (\ref{eq:lagrangianonesone}) can be rewritten as
\begin{equation}\label{eq:rewritelagrangianonecaseone}
\mathcal{L}_1 = \int_{p_{1,\ell}}^{p_{1,\ell}+p_{0,\ell}} u_{0,\ell}(x) \, dx + \int_0^{p_{1,\ell}} u_{1,\ell}(x) \,dx 
\end{equation}
and upper bounded as
\begin{equation}\label{eq:upperboundlagrangianonecaseone}
\mathcal{L}_1  \leq \int_0 ^{+\infty} \left[ \max \lbrace u_{0,\ell}(x),u_{1,\ell}(x)\rbrace \right]^+  \,dx .
\end{equation}
The root of $u_{0,\ell}(x)$ is $\gamma_{1,\ell}$ defined in (\ref{eq:intersectioncaseone}) while the largest root of $u_{1,\ell}(x)$ is $\beta_{1,\ell}$ defined in (\ref{eq:squareroot}). $u_{0,\ell}(x)$ and $u_{1,\ell}(x)$ intersect at the point $\zeta_{1,\ell}$ given by (\ref{eq:intersectioncaseone}). In the following, we consider two cases. 

\begin{enumerate}
\item $\frac{w_1}{w_0} > \frac{\alpha_{1,\ell}^-}{\alpha_{1,\ell}^{-} -\alpha_{2,\ell}^+}$, i.e., $\zeta_{1,\ell}$ is positive. 

In this case, $u_{1,\ell}(0)>u_{0,\ell}(0)$. There are three possibilities to consider depending on the value of $\lambda$.
\begin{enumerate}
\item If $u_{1,\ell}(0)<0$, then both $u_{0,\ell}(x)$ and $u_{1,\ell}(x)$ are negative for $x>0$, and (\ref{eq:upperboundlagrangianonecaseone}) is achieved by $p_{0,\ell}^{(1)}=0$ and $p_{1,\ell}^{(1)}=0$.
\item If $u_{1,\ell}(0) \geq 0$ and $\gamma_{1,\ell} < \zeta_{1,\ell}$, then  (\ref{eq:upperboundlagrangianonecaseone}) is achieved by $p_{0,\ell}^{(1)}=0$ and $p_{1,\ell}^{(1)}=\beta_{1,\ell}$.
\item If $\gamma_{1,\ell}  \geq \zeta_{1,\ell}$, then (\ref{eq:upperboundlagrangianonecaseone}) is achieved by $p_{0,\ell}^{(1)}=\gamma_{1,\ell} -\zeta_{1,\ell}$ and $p_{1,\ell}^{(1)}=\zeta_{1,\ell}$.
\end{enumerate} 
In summary, we obtain (\ref{p11}).

\item $\frac{w_1}{w_0} \leq \frac{\alpha_{1,\ell}^-}{\alpha_{1,\ell}^{-} - \alpha_{2,\ell}^+}$, i.e., $\zeta_{1,\ell}$ is negative.

In this case, $u_{0,\ell}(0) \geq u_{1,\ell}(0)$. 
\begin{enumerate}
\item If $u_{0,\ell}(0) \leq 0$, then (\ref{eq:upperboundlagrangianonecaseone}) is achieved by $p_{0,\ell}^{(1)}=0$ and $p_{1,\ell}^{(1)}=0$. 
\item If $u_{0,\ell}(0) > 0$, then (\ref{eq:upperboundlagrangianonecaseone}) is achieved by $p_{0,\ell}^{(1)}=\gamma_{1,\ell}$ and $p_{1,\ell}^{(1)}=0$.
\end{enumerate}
In summary, we obtain (\ref{p12}).
\end{enumerate}

For $\ell \in S_2$, $p_{0,\ell}^{(1)}$ and $p_{2,\ell}^{(1)}$ need to maximize :
\begin{align}\label{eq:lagrangiantwocaseone}
\mathcal{L}_2  = & \frac{w_0}{2} \log \bigl(1+\frac{\alpha_{1,\ell}^- p_{0,\ell}}{1+ \alpha_{1,\ell}^- p_{2,\ell}} \bigr) + \frac{w_2}{2}  \log \bigl( 1+ \alpha_{2,\ell}^- p_{2,\ell} \bigr) \nonumber \\
               & - \frac{w_2}{2}  \log \bigl(1+ \alpha_{1,\ell}^+ p_{2,\ell} \bigr)  - \lambda (p_{0,\ell} + p_{2,\ell}) .
\end{align} 
Then, we obtain, analogously to (\ref{eq:upperboundlagrangianonecaseone})
\begin{equation}\label{eq:upperboundlagrangiantwocaseone}
\mathcal{L}_2  \leq  \int_0 ^{+\infty} \left[ \max \bigl\{u_{0,\ell}(x),u_{2,\ell}(x)\bigr\} \right]^+  \,dx .
\end{equation}
The largest root of $u_{2,\ell}(x)$ is $\beta_{2,\ell}$ given by (\ref{eq:squareroot}). $u_{0,\ell}(x)$ and $u_{2,\ell}(x)$ intersect at two points. The largest point $\theta_{2,\ell}$ is given by (\ref{eq:intersectionnew}). We consider two cases depending on the sign of the two points.

\begin{enumerate}
\item $\frac{w_2}{w_0} > \frac{\alpha_{1,\ell}^-}{\alpha_{2,\ell}^{+}- \alpha_{1,\ell}^+}$, \textit{i.e.}, one point is negative and the other is positive.

In this case, $u_{2,\ell}(0)>u_{0,\ell}(0)$. There are three possibilities to consider.
\begin{enumerate}
\item If $u_{2,\ell}(0)<0$, then both $u_{0,\ell}(x)$ and $u_{2,\ell}(x)$ are negative for $x>0$, and (\ref{eq:upperboundlagrangiantwocaseone}) is achieved by $p_{0,\ell}^{(1)}=0$ and $p_{2,\ell}^{(1)}=0$.
\item If $u_{2,\ell}(0) \geq 0$ and $\gamma_{1,\ell} < \theta_{2,\ell}$, then (\ref{eq:upperboundlagrangiantwocaseone}) is achieved by $p_{0,\ell}^{(1)}=0$ and $p_{2,\ell}^{(1)}=\beta_{2,\ell}$. 
\item  If $\gamma_{1,\ell} \geq \theta_{2,\ell}$, then (\ref{eq:upperboundlagrangiantwocaseone}) is achieved by $p_{0,\ell}^{(1)}=\gamma_{1,\ell}-\theta_{2,\ell}$ and $p_{2,\ell}^{(1)}=\theta_{2,\ell}$.
\end{enumerate}
In summary, we obtain (\ref{p13}).

\item $\frac{w_2}{w_0} \leq \frac{\alpha_{1,\ell}^-}{\alpha_{2,\ell}^{+}- \alpha_{1,\ell}^+}$, \textit{i.e.}, the two intersection points are negative.

In this case, $u_{0,\ell}(0)\geq u_{2,\ell}(0)$. There are two possibilities to consider. 
\begin{enumerate}
\item If $u_{0,\ell}(0) \leq 0$, then (\ref{eq:upperboundlagrangiantwocaseone}) is achieved by $p_{0,\ell}^{(1)}=0$ and $p_{2,\ell}^{(1)}=0$.
\item If $u_{0,\ell}(0) > 0$, then (\ref{eq:upperboundlagrangiantwocaseone}) is achieved by $p_{0,\ell}^{(1)}=\gamma_{1,\ell}$ and $p_{2,\ell}^{(1)}=0$. 

In summary, we obtain (\ref{eq:powerallocationcaseone}).  
\end{enumerate}

The case that the two points are positive is not possible.
\end{enumerate}

For $\ell \in S_3$, $p_{0,\ell}^{(1)}$ need to maximize
\begin{equation}\label{eq:langrangianthree}
\mathcal{L}_3 = \frac{w_0}{2} \log \bigl (1 + \alpha_{1,\ell}^- p_{0,\ell} \bigr) - \lambda \, p_{0,\ell}.
\end{equation}
$\mathcal{L}_3$ can be upper bounded by
\begin{equation}\label{eq:upperboundlagrangianthree}
\mathcal{L}_3 = \int_0^{p_{0,\ell}} u_{0,\ell} (x) \,dx  \leq \int_0 ^{+ \infty} [ u_{0,\ell}(x) ]^+ \, dx .
\end{equation}
If $u_{0,\ell}(0) < 0$, then the upper bound on $\mathcal{L}_3$ is achieved by $p_{0,\ell}^{(1)}=0$. If $u_{0,\ell}(0) \geq 0$, the upper bound is achieved in this case by $p_{0,\ell}^{(1)} = \gamma_{1,\ell}$. In summary, we obtain (\ref{p18}).

The Lagrange parameter $\lambda$ is chosen to satisfy the power constraint with equality.

The solution of (P2) and (P3) is obtained in a similar fashion to that of (P1) and is not reported here for the sake of conciseness.

\bibliographystyle{IEEEtran}
\bibliography{biblio}

\end{document}